\documentclass[conference]{template/IEEEtran}
\IEEEoverridecommandlockouts
\usepackage{cite}
\usepackage{amsmath,amssymb,amsfonts,amsthm}
    \newcommand{\fnorm}[1]{\left|\left|\left| #1 \right| \right| \right|}
    \newcommand{\tr}{\mathrm{trace}}

    \renewcommand{\vec}{\mathbf}
    \newtheorem{theorem}{Theorem}
    \newtheorem{proposition}{Proposition}
\usepackage[shortlabels]{enumitem}
\usepackage{algorithmic}
\usepackage{hyperref}
\usepackage{graphicx}
\usepackage{textcomp}
\usepackage{xcolor}
\def\BibTeX{{\rm B\kern-.05em{\sc i\kern-.025em b}\kern-.08em
    T\kern-.1667em\lower.7ex\hbox{E}\kern-.125emX}}
\begin{document}

\title{The Maximal Eigengap Estimator for Acoustic Vector-Sensor Processing \\
\thanks{All authors acknowledge support from the Office of Naval Research, Grants No. N0001420WX01523, N0001421WX01634, and N0001421WX00410. We are also grateful to the Monterey Bay Aquarium Research Institute (MBARI) for their role in collecting the data used in this paper.}
}

\author{\IEEEauthorblockN{Robert Bassett}
\IEEEauthorblockA{\textit{Operations Research Department} \\
\textit{Naval Postgraduate School}\\
Monterey, CA 93943, USA\\
robert.bassett@nps.edu}\\
\IEEEauthorblockN{Paul Leary}
\IEEEauthorblockA{\textit{Physics Department} \\
\textit{Naval Postgraduate School}\\
Monterey, CA 93943, USA\\
pleary@nps.edu}
\and
\IEEEauthorblockN{Jacob Foster}
\IEEEauthorblockA{\textit{Operations Research Department} \\
\textit{Naval Postgraduate School}\\
Monterey, CA 93943, USA\\
jacob.foster@nps.edu}
\and
\IEEEauthorblockN{Kay L. Gemba}
\IEEEauthorblockA{\textit{Acoustics Division, Code 7160} \\
\textit{U.S. Naval Research Laboratory}\\
Washington, D.C. 20375, USA \\
kay.gemba@nrl.navy.mil} \\
\IEEEauthorblockN{Kevin B. Smith}
\IEEEauthorblockA{\textit{Physics Department} \\
\textit{Naval Postgraduate School}\\
Monterey, CA 93943, USA\\
kbsmith@nps.edu}
}

\maketitle

\begin{abstract}
This paper introduces the maximal eigengap estimator for finding the direction of arrival of a wideband acoustic signal using a single vector-sensor. We show that in this setting narrowband cross-spectral density matrices can be combined in an optimal weighting that approximately maximizes signal-to-noise ratio across a wide frequency band. The signal subspace resulting from this optimal combination of narrowband power matrices defines the maximal eigengap estimator. We discuss the advantages of the maximal eigengap estimator over competing methods, and demonstrate its utility in a real-data application using signals collected in 2019 from an acoustic vector-sensor deployed in the Monterey Bay.

\end{abstract}


\section{Introduction}
Direction of Arrival (DOA) estimation of acoustic signals is a problem which spans multiple application areas. Examples include biology, where tracking marine mammals can provide detailed information on their habits, and defense, where monitoring and locating vessels has utility in many naval applications. One advantage of an acoustic vector-sensor over a conventional hydrophone array is that acoustic vector-sensors have a smaller footprint, while still providing signal direction. The specialized nature of data collected by an acoustic vector-sensor motivates signal processing techniques customized for its analysis.

In this paper, we consider DOA estimation of a single wideband source using a single acoustic vector-sensor. We introduce the \emph{maximal eigengap estimator}, which combines cross-spectral density (CSD) matrices across a wide frequency range to maximize signal-to-noise ratio (SNR) of the resulting signal subspace. In the setting of a single acoustic vector sensor and a single source, we provide a tractable formulation of the maximal eigengap estimator, thus providing a new DOA estimation method customized for this setting.

Previous work on DOA estimation of wideband signals aggregates narrowband information across a wide frequency band in a variety of ways. In \cite{CSSM}, CSD matrices for each frequency are combined via linear combination, where the ideal weighting of each matrix is given by the SNR in that frequency. 
Because this SNR is unknown, the authors take an equal weighting of each cross-spectral matrix. The authors of \cite{WSF} instead aggregate narrowband frequency information by forming an estimating equation for the signal subspace which is a linear combination of generalized eigenvector equations for each frequency. It is suggested that the weighting should be related to SNR in each frequency (this suggestion is confirmed in an analysis of asymptotic variance \cite{OptWeights}), but numerical examples suggest that a uniform weighting works equally well. Other popular methods combine signal or noise subspace information uniformly \cite{TOPS, Mixed}. Our contribution, the maximal eigengap estimator, continues this pattern of aggregating narrowband information, and is similar in theme to \cite{CSSM}, with the important difference that we capitalize on the acoustic vector-sensor setting to \emph{optimize} SNR over the weights in a linear combination of cross-spectral power matrices.

This paper is organized as follows. In the next section, we formally introduce the maximal eigengap estimator and our main theoretical contribution, Theorem \ref{prop:convex_quad}, which provides a tractable reformulation of the estimator. Section \ref{sec:exper} applies the maximal eigengap estimator to signals collected in 2019 by an acoustic vector-sensor deployed in the Monterey Bay, where ground-truth information on vessel locations provide a realistic test case. We conclude by summarizing our results.

Before proceeding we establish some notation. We denote vectors and matrices by bold text, and scalars by plain text. For a vector or matrix $\vec{a}$, we denote its transpose by $\vec{a}^{T}$ and its conjugate transpose by $\vec{a}^{H}$. Similarly, we denote the complex conjugate of a scalar and pointwise conjugate of vector/matrix with superscript $^{*}$. Minimal and maximal eigenvalues of a Hermitian matrix $\vec{a}$ (which are real by the spectral theorem \cite{axler}) are denoted $\lambda_{\text{min}}(\vec{a})$ and $\lambda_{\text{max}}(\vec{a})$, with corresponding eigenvectors $\vec{v}_{\text{min}}(\vec{a})$ and $\vec{v}_{\text{max}}(\vec{a})$. The expectation operator is written $\mathbb{E}\left[\cdot\right]$. The Frobenius norm of a matrix $\vec{a}$ is denoted $\fnorm{\vec{a}}$, and $\|\cdot\|$ denotes an arbitrary norm.

\section{Maximal Eigengap Estimator}\label{sec:methodology}

\subsection{Signal Model}

Consider a single acoustic vector-sensor with a single signal source. The sensor's output at time $t$, $\vec{z}(t)$, has four channels, consisting of an omnidirectional hydrophone and three particle velocity measurements. We assume a plane wave signal, where it can be shown (see \cite{nehorai}) that the scaled sensor output is
\begin{equation}\label{time_domain1}
\vec{z}(t) = \vec{w} \, s(t) + \vec{n}(t),
\end{equation}
where $s(t)$ is the acoustic pressure at the sensor at time $t$, $\vec{w} = \left(1, k_{x}, k_{y}, k_{z} \right)^{T}$ is a vector such that $\vec{k} = \left(k_{x}, k_{y}, k_{z}\right)^{T}$ is a unit vector pointing towards the stationary source, and $\vec{n}(t) \in \mathbb{R}^{4}$ is a noise term. The vector $\vec{k}$ can be written $\left( \cos \theta \, \cos \phi, \sin \theta  \cos \phi, \sin \phi \right)^{T}$, where $\theta$ and $\phi$ give the azimuth and elevation angles, respectively, of $\vec{k}$. Throughout, we focus our attention on estimating the azimuth angle $\theta$ using the $x$ and $y$ velocity channels, setting $\phi=0$. This reflects that azimuth is the primary quantity of interest in many DOA estimation problems. 

Denote by $\vec{x}(t)$, $\vec{e}(t)$, and $\vec{u}$, the restriction of $\vec{z}(t)$, $\vec{n}(t)$, and $\vec{w}$, respectively, to their second ($x$) and third ($y$) channels. Because we assume $\phi = 0$, $\vec{u}$ is a unit vector. Equation \eqref{time_domain1} restricted to $x$ and $y$ velocity channels is then
\begin{equation}\label{time_domain2}
\vec{x}(t) = \vec{u} \, s(t) + \vec{e}(t).
\end{equation}

A frequency domain representation of \eqref{time_domain2} is
\begin{equation}\label{freq_domain}
\vec{X}(\omega) = \vec{u} \, S(\omega) + \vec{E}(\omega).
\end{equation}

Denote $\mathbb{E}\left[ S(\omega)^2 \right]$ by $P_{S}(\omega)$, and $\mathbb{E}\left[\vec{E}(\omega) \,  \vec{E}(\omega)^{H}\right]$ by $\mathbf{\Sigma}(\omega)$. Assume that 
\begin{enumerate}[label=(A\arabic*)]
\item \label{A1} $\mathbb{E}\left[S(\omega) \,  \vec{E}(\omega)^{*}\right] = \mathbb{E}\left[S(\omega)^{*} \,  \vec{E}(\omega)\right] = 0$ for each $\omega$.
\item \label{A2} The condition number of $\vec{\Sigma}(\omega)$ is bounded by some constant $C$, uniformly in $\omega$.
$$\frac{\lambda_{\text{max}}\left(\vec{\Sigma}(\omega)\right)}{\lambda_{\text{min}}\left(\vec{\Sigma}(\omega)\right)} \leq C$$
\end{enumerate} 
Assumption \ref{A1} allows us to form the CSD matrix of $\vec{X}(\omega)$, denoted $\vec{P}_{\vec{X}}(\omega)$, as
\begin{align}
\vec{P}_{\vec{X}}(\omega) &= \mathbb{E}\left[\vec{X}(\omega) \, \vec{X}(\omega)^{H}\right] \nonumber\\
 &= P_{S}(\omega) \,\vec{u} \,\vec{u}^{T} + \mathbf{\Sigma}(\omega). \label{covar}
\end{align} 

When $P_{S}(\omega)$ dominates $\vec{\Sigma}(\omega)$, it can be shown that the maximal eigenvector $\vec{v}_{\text{max}}\left(\vec{P}_{\vec{X}}(\omega)\right)$ is close to $\vec{u}$. Moreover, if the spatial covariance $\vec{\Sigma}(\omega)$ is a scalar multiple of the identity matrix this recovery is exact. A precise statement is given by the following proposition, which follows from a direct application of \cite[Theorem 8.5]{Wainwright}.

\begin{proposition}
Assume $P_{S}(\omega) >0$ and $\fnorm{\vec{\Sigma}(\omega)} < P_{S}(\omega)/2$. Let $\vec{u}^{\perp}$ be a unit vector perpendicular to $\vec{u}$, $\vec{U} = \left[\vec{u}, \vec{u}^{\perp}\right]$ a matrix with columns $\vec{u}$ and $\vec{u}^{\perp}$, and $\tilde{p}(\omega)$ the off-diagonal term in $\vec{U}^{T} \, \vec{\Sigma}(\omega) \, \vec{U}$. Then \begin{equation}\label{prop_eq}
\left\|\vec{v}_{\text{max}}\left(\vec{P}_{\vec{X}}(\omega)\right) - \vec{u}\right\|_{2} \leq \frac{2 \left|\tilde{p}(\omega)\right|}{P_{S}(\omega) - 2\fnorm{\vec{\Sigma}(\omega)}}.
\end{equation}
\end{proposition}

By the bound given in \eqref{prop_eq}, we have tighter control on the difference of $\vec{v}_{\text{max}}\left(\vec{P}_{\vec{X}}(\omega)\right)$ and the signal's DOA $\vec{u}$ when either
\begin{enumerate}[(i)]
\item \label{item1} The maximum eigenvector of $\Sigma(\omega)$ is closely aligned with $\vec{u}$, in the sense that $\left| \vec{v}_{\text{max}}(\vec{\Sigma}(\omega))^{T} \, \vec{u}^{\perp}\right|$ (which directly controls $\left|\tilde{p}(\omega)\right|$) is small.
\item \label{item2} The signal dominates the noise, in the sense that $P_{S}(\omega)$ is large and $\vec{\Sigma}(\omega)$ is small (in Frobenius norm).
Recalling that $\fnorm{\vec{\Sigma}(\omega)} = \sqrt{\lambda_{\text{max}}\left(\vec{\Sigma}(\omega)\right)^2 + \lambda_{\text{min}}\left(\vec{\Sigma}(\omega)\right)^2}$, we can alternatively insist that these eigenvalues are small.
\end{enumerate}

Using assumption \ref{A2}, we can upper bound $\fnorm{\vec{\Sigma}(\omega)}$ as
$$\fnorm{\vec{\Sigma}(\omega)} \leq \lambda_{\text{min}}\left(\vec{\Sigma}(\omega)\right)\sqrt{1 + C^2},$$
and the denominator in \eqref{prop_eq} 
\begin{equation}\label{denominator}
P_{S}(\omega) - 2\fnorm{\vec{\Sigma}(\omega)} \geq P_{S}(\omega) - 2 \lambda_{\text{min}}\left(\vec{\Sigma}(\omega)\right)\sqrt{1 + C^2}.
\end{equation}
From \eqref{prop_eq} and \eqref{denominator}, we see that large $P_{S}(\omega)$ and small $\lambda_{\text{min}}\left(\vec{\Sigma}(\omega)\right)$ make $\vec{v}_{\text{max}}\left(\vec{P}_{\vec{X}}(\omega)\right)$ a better estimate of $\vec{u}$. However, note that both $P_{S}(\omega)$ and $\lambda_{\min}(\Sigma(\omega))$ are properties of unobserved variables. From \eqref{covar}, we note that $\lambda_{\text{max}}(\vec{P}_{\vec{X}}(\omega))$ approximates $P_{S}(\omega)$ and $\lambda_{\text{min}}(\vec{P}_{\vec{X}}(\omega))$ approximates $\lambda_{\text{min}}(\vec{\Sigma}(\omega))$. Hence our desire for large $P_{S}(\omega)$ and small $\lambda_{\text{min}}\left(\vec{\Sigma}(\omega)\right)$ can be stated as the following. 

\textbf{Goal}: for better estimation of the signal DOA, we should have $\lambda_{\text{max}}(\vec{P}_{\vec{X}}(\omega))$ large and $\lambda_{\text{min}}(\vec{P}_{\vec{X}}(\omega))$ small.

\subsection{Combining Narrowband Information} 

For wideband signals, we seek to aggregate the narrowband information described in the previous section over a range of $\omega$ values, only some of which contain the signal of interest. For frequency bins $F = \{\omega_{1}, ..., \omega_{|F|}\}$, our goal is to decide, for each $\omega \in F$, whether the signal of interest is contained in the narrowband CSD matrix $\vec{P}_{\vec{X}}(\omega)$, and if so to what degree it agrees with the signal DOA in other frequency bins. To accomplish this goal, we extend \eqref{covar} by searching for a weighted combination of the power matrices that gives a tighter bound in \eqref{prop_eq}. Namely, we propose solving, for some to-be-specified norm $\|\cdot\|$, the following problem, where $\vec{a} \in \mathbb{R}^{|F|}$ is a weight vector such that $a_{\omega}$ denotes the entry of $\vec{a}$ corresponding to frequency bin $\omega \in F$.
\begin{align}
\max_{\vec{a} \in \mathbb{R}^{|F|}} \; & \lambda_{\text{max}}\left(\sum_{\omega \in F} a_{\omega} \vec{P}_{\vec{X}}(\omega) \right) -\lambda_{\text{min}} \left(\sum_{\omega \in F} a_{\omega} \vec{P}_{\vec{X}}(\omega) \right) \label{og_prob}\\
& \quad \quad \text{s.t.}  \quad  \|\vec{a}\| \leq 1, \quad \vec{a} \geq 0 \nonumber
\end{align}
In \eqref{og_prob}, the $\|\cdot\| \leq 1$ constraint defines the maximum size of the weight vector. Because the $\lambda_{\text{max}}$ term measures the power in the signal subspace, and $\lambda_{\text{min}}$ term the power in the noise subspace, this objective maximizes the estimated SNR\footnote{In detail, maximizing difference in \eqref{og_prob} is equivalent to maximizing $\log(\lambda_{\text{max}}/\lambda_{\text{min}})$ which is equivalent to maximizing the SNR $\lambda_{\text{max}}/\lambda_{\text{min}}$.}.

One concern with the formulation \eqref{og_prob} is that some frequencies (in practice, often the lower ones) have more power in both signal and noise subspaces, so these terms have the potential to dominate the objective function. To remedy this, we propose three standardization schemes for the CSD matrices $\vec{P}_{\vec{X}}(\omega)$. We denote a standardized $\vec{P}_{\vec{X}}(\omega)$ by $\vec{Q}_{\vec{X}}(\omega)$. Three intuitive options for standardization are:
\begin{enumerate}[(a)]
\item Take each matrix to have unit trace, so that \label{scaling1}
$$\vec{Q}_{\vec{X}}(\omega) = \frac{\vec{P}_{\vec{X}}(\omega)}{\tr\left(\vec{P}_{\vec{X}}(\omega)\right)},$$
and the total power (aggregated across both signal and noise subspaces) in each frequency bin is 1.
\item Take each matrix to have minimal eigenvalue 1, so that 
$$\vec{Q}_{\vec{X}}(\omega) = \frac{\vec{P}_{\vec{X}}(\omega)}{\lambda_{\text{min}}\left(\vec{P}_{X}(\omega)\right)},$$
and the maximal eigenvalue of the matrix gives its SNR.
\item No standardization, so that
$$\vec{Q}_{\vec{X}}(\omega) = \vec{P}_{\vec{X}}(\omega).$$ \label{scaling3}
\end{enumerate} 
We compare the merits of each standardization method in section \ref{sec:exper}.

Lastly, we replace $\vec{P}_{\vec{X}}(\omega)$ in \eqref{og_prob} with an estimate $\hat{\vec{P}}_{\vec{X}}(\omega)$ obtained from observations. We extend the notation for standardized matrices as one would expect, so that $\hat{\vec{Q}}_{\vec{X}}(\omega)$ is the standardized $\hat{\vec{P}}_{\vec{X}}(\omega)$. Then \eqref{og_prob} becomes
\begin{equation}\label{eq5}
\max_{\vec{a} \in \mathbb{R}^{|F|}} \lambda_{\text{max}}\left(\sum_{\omega \in F} a_{\omega} \hat{\vec{Q}}_{\vec{X}}(\omega) \right) -   \lambda_{\text{min}}\left(\sum_{\omega \in F} a_{\omega} \hat{\vec{Q}}_{\vec{X}}(\omega)  \right)
\end{equation}
$$\text{s.t.} \quad \|\vec{a}\| \leq 1, \quad \vec{a} \geq 0.$$

We define the \emph{maximal eigengap estimator} as the real part of a maximum eigenvector, $\text{Re}\left\{\vec{v}_{\text{max}}\left(\sum_{\omega \in F} \hat{a}_{\omega} \hat{\vec{Q}}_{\vec{X}}(\omega) \right)\right\}$, where $\hat{\vec{a}}$ denotes the maximizer in \eqref{eq5}. The following Theorem is instrumental in computing the maximal eigengap estimator.

\begin{theorem}\label{prop:convex_quad}
The objective in \eqref{eq5} is the square root of a convex quadratic function in $\vec{a}$.
\end{theorem}
\begin{proof}
Denote the entries of $\hat{\vec{Q}}_{\vec{X}}(\omega)$ by
$$\hat{\vec{Q}}_{\vec{X}}(\omega) = \left(\begin{array}{cc} q_{\omega} & r_{\omega}\\ r^{*}_{\omega} & s_{\omega} \end{array}\right),$$
where we have used the conjugate symmetry of $\hat{\vec{Q}}_{\vec{X}}(\omega)$ in the off-diagonal terms. The eigenvalues 
$$\lambda_{\text{max}}\left(\sum_{\omega \in F} a_{\omega} \hat{\vec{Q}}_{\vec{X}}(\omega) \right) \quad \text{ and } \quad  \lambda_{\text{min}}\left(\sum_{\omega \in F} a_{\omega} \hat{\vec{Q}}_{\vec{X}}(\omega)  \right)$$
are given by the zeros of the characteristic polynomial for $\sum_{\omega \in F} a_{\omega} \hat{\vec{Q}}_{\vec{X}}(\omega)$,
\begin{equation}\label{det_is_zero}
\det\left(\begin{array}{cc} \sum_{\omega \in F} a_{\omega} q_{\omega} - \lambda & \sum_{\omega \in F} a_{\omega} r_{\omega} \\ \sum_{\omega \in F} a_{\omega} r^{*}_{\omega} & \sum_{\omega \in F} a_{\omega} s_{\omega} - \lambda \end{array} \right) = 0.
\end{equation}
These roots can be computed explicitly using the quadratic formula, and their difference is the discriminant of the quadratic in \eqref{det_is_zero}, which is nonnegative because the eigenvalues of Hermitian matrices are real.
\begin{equation}\label{discriminant}
\scriptsize
\sqrt{\left(\sum_{\omega \in F} a_{\omega}\left( q_{\omega} + s_{\omega}\right)\right)^{2} - 4 \left(\left(\sum_{\omega \in F} a_{\omega} q_{\omega} \right) \left( \sum_{\omega \in F} a_{\omega} s_{\omega}\right) - \left| \sum_{\omega \in F} a_{\omega} r_{\omega} \right|^2\right)}
\end{equation}
Because we maximize this expression in \eqref{eq5}, and it is always nonnegative, we can instead maximize the square of \eqref{discriminant}
\begin{equation}\label{square_discriminant}
\scriptsize
\left(\sum_{\omega \in F} a_{\omega} \left(q_{\omega} + s_{\omega} \right)\right)^{2} - 4 \left(\left(\sum_{\omega \in F} a_{\omega} q_{\omega} \right) \left( \sum_{\omega \in F} a_{\omega} s_{\omega}\right) - \left| \sum_{\omega \in F} a_{\omega} r_{\omega} \right|^2\right).
\end{equation}
Expanding \eqref{square_discriminant}, we have
\begin{align}\label{quad_form}
\sum_{\omega_{i} \in F} \sum_{\omega_{j} \in F} \bigg( & a_{\omega_{i}} \left(q_{\omega_{i}} + s_{\omega_{i}}\right) \left(q_{\omega_{j}} + s_{\omega_{j}}\right) a_{\omega_{j}} \\
& - 4 \left( a_{\omega_{i}} q_{\omega_{i}} s_{\omega_{j}} a_{\omega_{j}} - a_{\omega_{i}} r_{\omega_{i}} r^{*}_{\omega_{j}} a_{\omega_{j}} \right) \bigg) \nonumber
\end{align}

Define a matrix $\tilde{\vec{R}} \in \mathbb{C}^{|F| \times |F|}$ with
\begin{equation}\label{R_tilde}
\tilde{\vec{R}}_{\omega_{i}, \omega_{j}} = \left(q_{\omega_{i}} + s_{\omega_{i}}\right) \left(q_{\omega_{j}} + s_{\omega_{j}}\right) - 4\left( q_{\omega_{i}} s_{\omega_{j}} - r_{\omega_{i}} r^{*}_{\omega_{j}}\right),
\end{equation}
so that the quadratic form $\vec{a}^{T} \, \tilde{\vec{R}} \, \vec{a}$ gives \eqref{quad_form}. 

Next we apply a common technique to produce a real symmetric quadratic form $\vec{R}$ which is equal to the quadratic form given by $\tilde{\vec{R}}$. First, note that expression \eqref{square_discriminant} is \emph{real} and \emph{scalar}, because $a_{\omega}$, and $q_{\omega}$, and $s_{\omega}$ are real (recall that $\hat{\vec{Q}}_{\vec{X}}(\omega)$ is Hermitian). Hence \eqref{square_discriminant} is equal to its conjugate, its transpose, and its conjugate transpose. Because \eqref{quad_form} and \eqref{square_discriminant} are equal, it follows that
\begin{equation}\label{long_eq}
\left(\vec{a}^{T} \tilde{\vec{R}} \vec{a}\right)^{T} = \left(\vec{a}^{T} \tilde{\vec{R}} \vec{a}\right)^{H} = \vec{a}^{T} \tilde{\vec{R}}^{*}\vec{a} = \vec{a}^{T} \tilde{\vec{R}} \vec{a}.
\end{equation}
Define 
\begin{equation}\label{R_def}
\vec{R} = \frac{\frac{1}{2}\left(\tilde{\vec{R}} + \tilde{\vec{R}}^{H}\right)+ \frac{1}{2}\left(\tilde{\vec{R}} + \tilde{\vec{R}}^{H}\right)^{T}}{2},
\end{equation}
which is real, symmetric, and from \eqref{long_eq} yields the same quadratic form as $\tilde{\vec{R}}$.

Finally, we remark that the quadratic form induced by $\vec{R}$ is convex because $\vec{a}^{T} \vec{R} \vec{a}$ is nonnegative for any choice of $a$. This follows from the derivation of \eqref{discriminant}, in which we note that the discriminant is nonnegative because the eigenvalues of the conjugate symmetric matrix $\sum_{a \in F} a_{\omega} \hat{\vec{Q}}_{\vec{X}}(\omega)$ must be real. 
\end{proof}

Theorem \ref{prop:convex_quad} permits an equivalent definition of the weights $\hat{\vec{a}}$ in the maximal eigengap estimator of \eqref{eq5} as the maximizer of the following expression, where $\vec{R}$ is defined in \eqref{R_def}.
\begin{align}\label{main_prob}
& \max_{\vec{a} \in \mathbb{R}^{|F|}}  \vec{a}^{T} \vec{R} \vec{a}\\
\text{s.t.} & \quad \|\vec{a}\| \leq 1 \quad \vec{a} \geq 0. \nonumber
\end{align}
The primary benefit of this reformulation is the simple form of problem \eqref{main_prob}, which facilitates computation. We consider using both the ${2}$-norm $\|\cdot\|_{2}$ and the ${1}$-norm $\|\cdot\|_{1}$ to define the weight constraint. Note that, regardless of the norm used, \eqref{main_prob} maximizes a convex function over a convex set, so a maximizer exists and is contained within the extreme points of the feasible set \cite[Corollary 32.3.4]{rockafellar}.

\section{Applications \& Experiments}\label{sec:exper}

In this section we apply the maximal eigengap estimator to a set of signals collected in 2019 by an acoustic vector-sensor deployed in the Monterey Bay. We combine these signals with Automatic Identification System (AIS) data providing the GPS locations of vessels in the Bay throughout 2019. By comparing actual vessel DOAs with estimates provided by the maximal eigengap estimator, we assess the performance of the maximal eigengap estimator in a realistic setting.

\subsection{Data Details}\label{sec:data}

The acoustic signals we consider were collected by a Geospectrum M20-105 vector-sensor deployed at a depth of 891 meters on the Monterey Accelerated Research System cabled observatory, which is operated by the Monterey Bay Aquarium Research Institute. The time window of the data spans from February 1st to December 31st, 2019, and except for a handful of minor, maintenance-related disruptions these data are a continuous representation of the acoustic soundscape in Monterey Bay over the time frame considered.

\begin{figure}[b]
\begin{center}
\includegraphics[width=.5\textwidth]{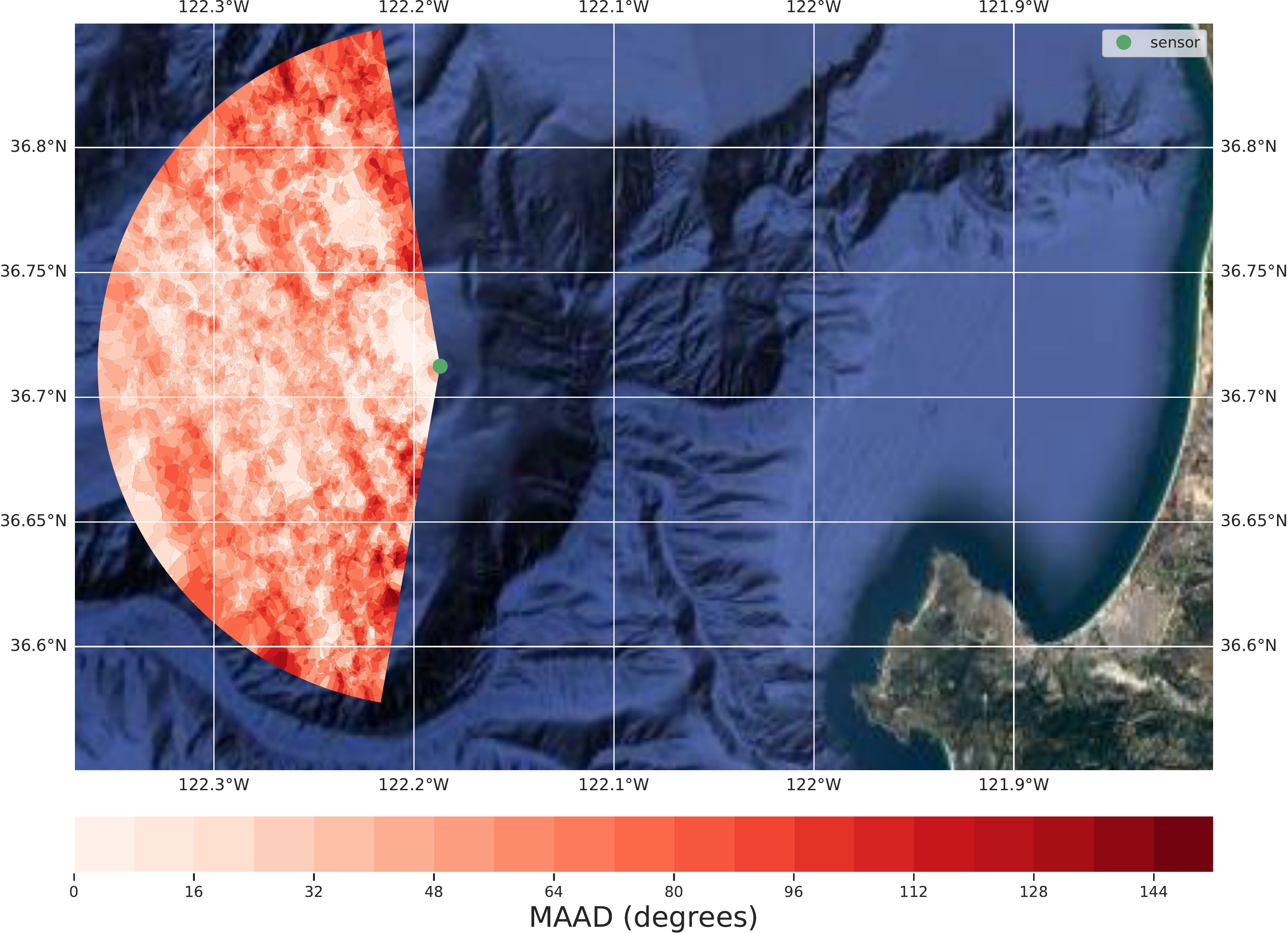}
\end{center}
\caption{The acoustic vector-sensor within the Monterey Bay overlaid with an interpolation of the mean absolute angular deviation for the maximal eigengap estimator (with ${2}$-norm and no scaling, see Table \ref{tab:methods}) constructed using K-nearest neighbors regression over all observed vessels.}
\label{fig:map}
\end{figure}

The locations of vessels are given by GPS data provided by the US Coast Guard, with a time resolution of five minutes. We consider all observations where a vessel is due West, within 15 km, and has bearing between $190^{\circ}$ and $350^{\circ}$ of the sensor. We also require that there are no additional vessels identified within 20 km of the sensor. We pair each of these vessel locations with the corresponding five minutes of signal recorded by the acoustic vector-sensor, resulting in 3674 observations labeled by signal DOA. By restricting our attention to vessels West of the sensor, we focus on detecting an eigenvector defining the signal subspace, and not considering vessels within $10^{\circ}$ of North or South avoids complications arising from the ambiguity of its sign. We note that once a signal subspace has been detected, additional processing can be used to resolve this ambiguity using the acoustic vector-sensor's omnidirectional channel \cite{nehorai}.

Figure \ref{fig:map} gives the location of the vector-sensor within the Monterey Bay, with an interpolated error function for the estimator superimposed. Bathymetrically, the sensor is located on a shelf within the Monterey Canyon. To the sensor's East, the ocean is shallower and contains mostly fishing and recreational vessels, whereas West of the sensor a pair of nearby shipping channels yields traffic that is primarily commercial. The data contain an unknown proportion of errors caused by noisy GPS reports and interfering signals such as small recreational vessels or aquatic mammals. The multi-season nature of the data also produces dynamic propagation conditions which impact the strength of both source and interfering signals at the sensor \cite{kay}.

\subsection{Numerical Performance}

To assess the maximal eigengap estimator's performance, we apply it to the acoustic signals with known DOAs described in \ref{sec:data}. We take frequency bins $F$ ranging from 75 Hz to 300 Hz with 2 Hz resolution, and use averaged periodograms to estimate the CSD matrices $\hat{\vec{P}}_{\vec{X}}(\omega)$ \cite{spectral_est}. In \ref{sec:methodology}, we introduced two opportunities for variation in the implementation of the maximal eigengap estimator \eqref{main_prob}, depending on the norm used to constrain the weight vector and the scaling of the CSD matrix included in the data preprocessing \ref{scaling1}-\ref{scaling3}. We consider three of these variations in the following experiments: ${1}$-norm with trace scaling, ${2}$-norm with minimal eigenvalue scaling, and ${2}$-norm with no scaling. As a comparison of how the maximal eigengap estimator performs relative to existing methods, we also implement and apply the covariance-based DOA estimator from \cite{nehorai}. Table \ref{tab:methods} gives the different estimators considered.

\begin{table}
\begin{center}
\begin{tabular}{|c|c|c|c|} \hline
Reference & Abbreviation & Norm & Scaling\\\hline
This paper & ${1}$-Trace & ${1}$-norm & Trace\\
This paper & ${2}$-MinEig & ${2}$-norm & Minimum Eigenvalue\\
This paper & ${2}$-None & ${2}$-norm & None\\
\cite{nehorai} & Covar & NA & NA \\ \hline
\end{tabular}
\end{center}
\caption{The estimators considered. The maximal eigengap \eqref{main_prob} uses various norms and CSD matrix scalings. The covariance estimator of \cite{nehorai} is included as a competing method.}
\label{tab:methods}
\vspace{-1cm}
\end{table}

We comment briefly on the implementation details of the maximal eigengap estimator using the ${1}$ and ${2}$-norms. For the ${1}$-norm constraint, the weight vector in the maximal eigengap estimator can be computed in closed form. Indeed, when the 1-norm is used in \eqref{main_prob}, the extreme points of the constraint set are $\left\{\mathbf{0}, \mathbf{e}_{1} ,...,\mathbf{e}_{|F|}\right\}$, where $\mathbf{0}$ denotes the zero vector and $\mathbf{e}_{i}$ the $i$th standard basis vector. Then it is clear that the maximizer in \eqref{main_prob} occurs at the standard basis vector corresponding to the largest diagonal term in $\mathbf{R}$. 

\begin{figure}[b]
\vspace{-.5cm}
\begin{center}
\includegraphics[width=.5\textwidth]{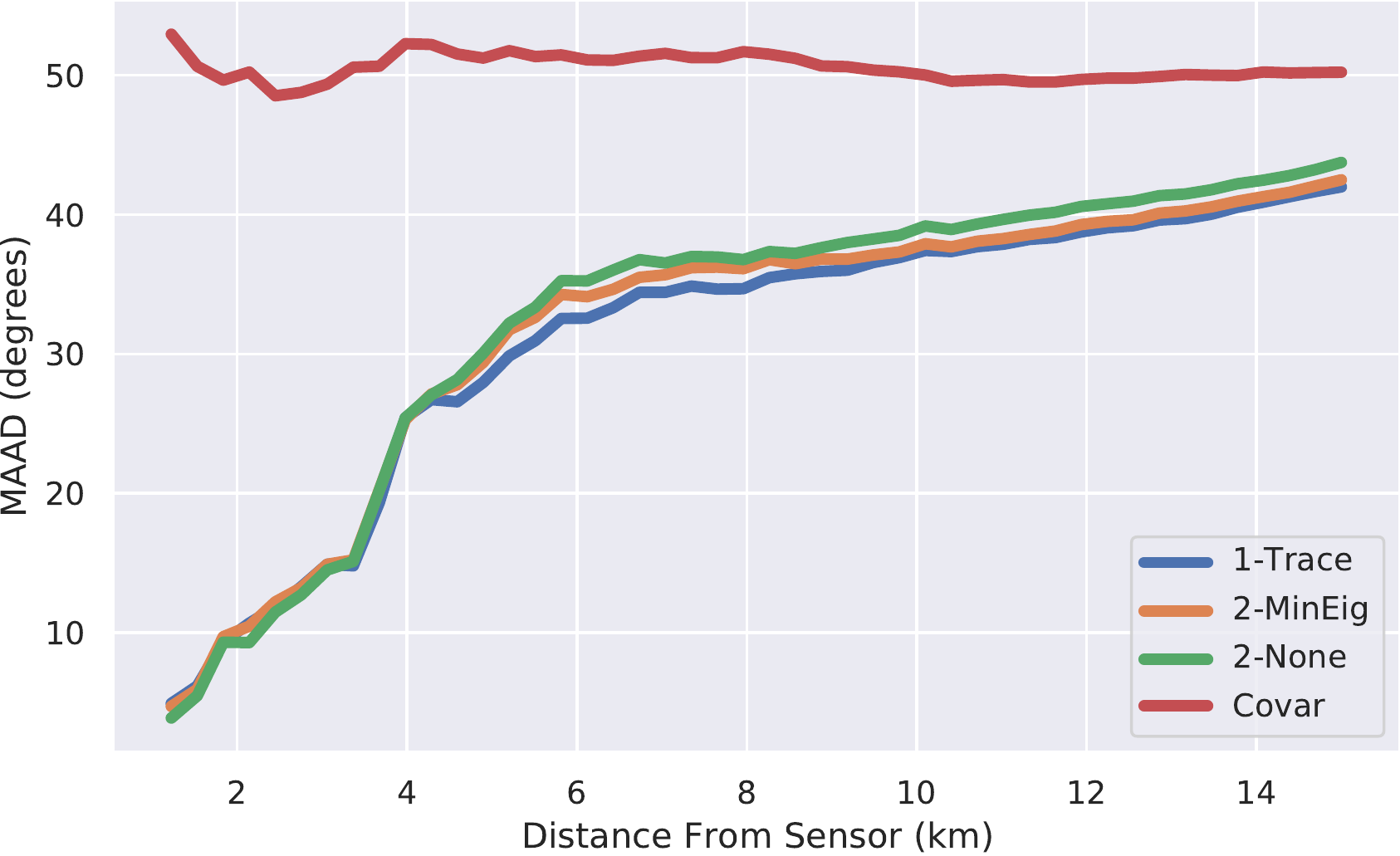}
\end{center}
\caption{Mean absolute angular deviation for each DOA estimator in Table \ref{tab:methods}, as a function of the distance between a vessel and the sensor. The three variants of the maximal eigengap estimator perform much better than the covariance method of \cite{nehorai} when the vessels are closer to the sensor, and this advantage decreases as the distance increases.}
\label{fig:line}
\end{figure}

In the ${2}$-norm constrained case, \eqref{main_prob} resembles a maximal eigenvalue problem, but with an additional nonnegativity constraint. When the entries of $\mathbf{R}$ are nonnegative, the Perron-Frobenius theorem guarantees that $\mathbf{v}_{\text{max}}(\mathbf{R})$ satisfies the nonnegativity constraint, but from its construction $\mathbf{R}$ may have negative entries. Instead, we approximate a solution to \eqref{main_prob} by projecting $\mathbf{v}_{\text{max}}(\mathbf{R})$ onto the nonnegative orthant. In practice, we find that the $\mathbf{v}_{\text{max}}(\mathbf{R})$ is primarily composed of positive entries, and that occasional negative entries are close to zero, which suggests that this approximation is reasonable.



The results of these experiments demonstrate that the maximal eigengap estimator is a more accurate method for DOA estimation of an wideband signal using an acoustic vector-sensor than those existing in the literature. Figure \ref{fig:line} gives the mean absolute angular deviation (MAAD) of the various methods as a function of the upper bound on the distance between the sensor and vessels. Though variations of the maximal eigengap estimator perform similarly, the ${2}$-MinEig variation performs strictly worse than the ${1}$-Trace and ${2}$-None variations over all distances. Most importantly, all variations of the maximal eigengap estimator outperforms its competitor, the covariance method introduced in \cite{nehorai}. This performance difference is especially large for vessels close to the sensor, where the mean absolute angular deviation is approximately 30 degrees less for the maximal eigengap estimator. We conjecture that this difference is due to the flexibility of the maximal eigengap estimator, which selects from among a set of frequency bins those that present a similar signal subspace. The optimal frequency bins change depending on the vessel's distance from the sensor, and the maximal eigengap estimator has the ability to adapt to this change using its optimal weighting of CSD matrices.

\begin{figure}[ht]
\begin{center}
\includegraphics[width=.5\textwidth]{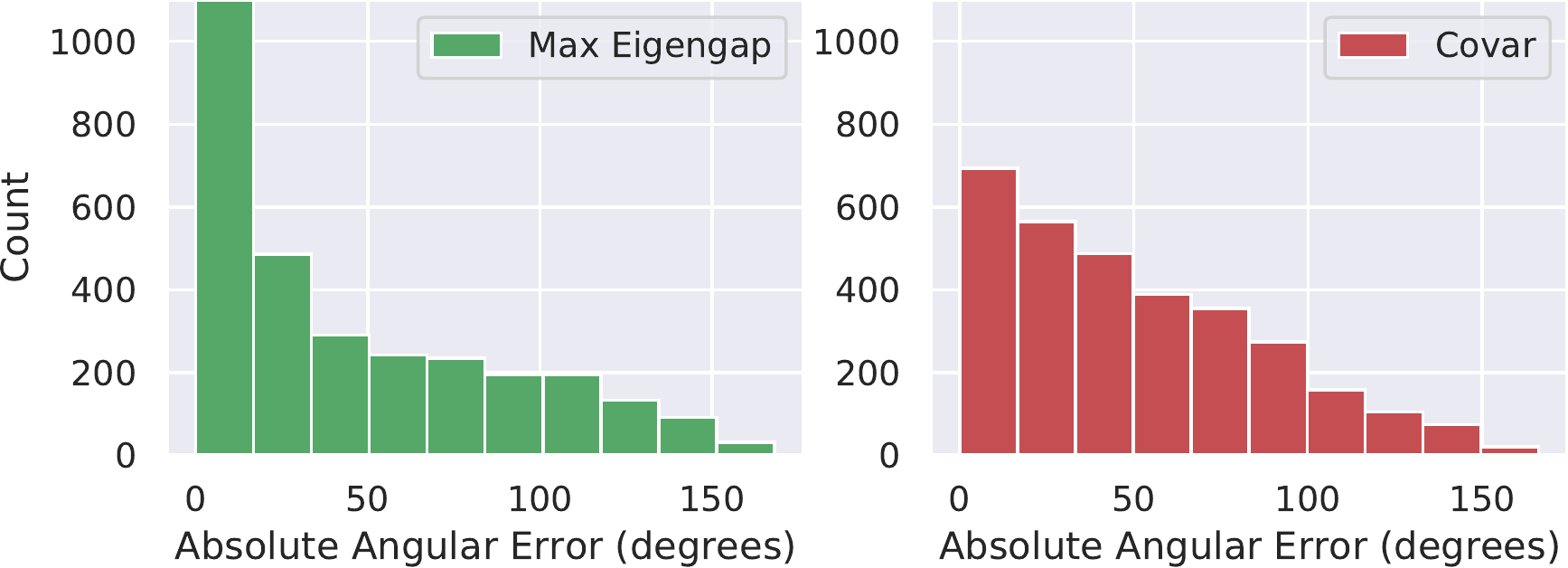}
\end{center}
\vspace{-.3cm}
\caption{Histograms of absolute angular error for the maximal eigengap estimator (with $2$-norm and no scaling) and the covariance method of \cite{nehorai}, applied to all vessels within 15 km of the sensor. The maximal eigengap estimator has error near zero more often than its competitor.}
\label{fig:hist}
\vspace{-.3cm}
\end{figure}

Further comparisons of the estimators' performance demonstrate the utility of the maximal eigengap estimator. Figure \ref{fig:hist} presents histograms of the absolute angular error for the ${2}$-None variant of the maximal eigengap estimator and the competing covariance method. In this figure, the maximal eigengap has $50\%$ more observations in the smallest error bin than the covariance method. Figure \ref{fig:map} presents the absolute angular deviation of the maximal eigengap estimator as a function of vessel location, constructed by interpolating error over all vessels in the data set. Certain vessel locations present more difficulty for DOA estimation than others. We suspect that this difficulty can primarily be attributed to propagation conditions arising from bathymetric features of those locations.

\section{Conclusion}

In this paper we introduce the maximal eigengap estimator for DOA estimation of a wideband signal collected with a single acoustic vector-sensor. The maximal eigengap estimator's utility is in its formulation, which optimally and tractably combines signal subspace information across a frequency range. We demonstrate that the maximal eigengap estimator outperforms existing techniques for DOA estimation of maritime vessels on a set of labeled data collected by an acoustic vector-sensor.

\bibliographystyle{template/IEEEtran}
\bibliography{template/IEEEabrv,thebib.bib}

\end{document}